\documentclass[11pt,a4paper,reqno]{amsart}
\usepackage[T1]{fontenc}
\usepackage[utf8]{inputenc}
\usepackage{lmodern}
\usepackage{amsmath,amssymb,amsthm,mathtools}
\usepackage[margin=26mm,headheight=14pt,footskip=18pt]{geometry}
\usepackage{microtype}
\usepackage{enumitem}
\usepackage{booktabs,tabularx,array}
\usepackage{xcolor}
\usepackage{fancyhdr}
\usepackage[unicode,pdfencoding=auto,psdextra]{hyperref}
\hypersetup{colorlinks=true,linkcolor=blue!45!black,citecolor=blue!45!black,
 urlcolor=blue!45!black,pdftitle={A generalized quantum Stein lemma for convex families on von Neumann algebras},
 pdfauthor={Li Gao},
 pdfsubject={Composite hypothesis testing and relative entropy on von Neumann algebras}}
\hypersetup{bookmarksdepth=2}
\setlist[enumerate]{topsep=4pt,itemsep=3pt,parsep=0pt}
\allowdisplaybreaks[2]
\numberwithin{equation}{section}
\newtheorem{theorem}{Theorem}[section]
\newtheorem{proposition}[theorem]{Proposition}
\newtheorem{lemma}[theorem]{Lemma}
\newtheorem{corollary}[theorem]{Corollary}
\theoremstyle{definition}

\newtheorem{assumption}[theorem]{Assumption}
\newtheorem{example}[theorem]{Example}
\theoremstyle{remark}
\newtheorem{remark}[theorem]{Remark}
\newcommand{\M}{\mathcal M}
\newcommand{\N}{\mathcal N}
\newcommand{\F}{\mathcal F}
\newcommand{\Sn}{\mathcal S}
\newcommand{\one}{\mathbf 1}

\newcommand{\dd}{\,\mathrm d}
\newcommand{\Tr}{\operatorname{Tr}}
\newcommand{\supp}{\operatorname{s}}
\newcommand{\co}{\operatorname{conv}}
\newcommand{\hock}{\mathsf H}
\newcommand{\Q}{\mathsf Q}
\newcommand{\Dht}{D_{\mathrm H}}
\newcommand{\norm}[1]{\left\lVert #1\right\rVert}

\newcommand{\rhon}{\rho^{\otimes n}}

\newcommand{\eref}[1]{\eqref{#1}}
\newcommand{\arxiv}[1]{\href{https://arxiv.org/abs/#1}{arXiv:#1}}
\newcommand{\doi}[1]{\href{https://doi.org/#1}{doi:#1}}

\title[GQSL for von Neumann algebras]{A Generalized quantum Stein lemma
on von Neumann algebras}
\author{Li Gao}\address{School of Mathematics and Statistics, Wuhan University, Wuhan 430072, China}
\address{Wuhan Institute of Quantum Technology, Wuhan 430075, China}
\email{gao.li@whu.edu.cn}
\date{\today}

\begin{document}

\begin{abstract}
We prove a generalized quantum Stein lemma for i.i.d. normal states against convex, tensor-stable families on arbitrary von Neumann algebras. Assuming the existence of an alternative state with finite relative entropy from the null state, we show that, at every type-I error tolerance $\varepsilon\in (0,1)$, the worst case type-II error exponent is achieved with the regularized relative entropy with a strong converse. The proof combines the integral representation of relative entropy by hockey-stick divergences, a modular testing bound, and convex minimax argument. 
\end{abstract}
\maketitle

{\bf \noindent Statement of AI usage.}
The solution and proof strategy presented in this paper were suggested by ChatGPT 6 Astra. The mathematical arguments were independently verified by the author, who takes full responsibility for the content of the manuscript. AI assistance was also used for the literature review and for writing, proofreading, and revising the manuscript.

\begingroup
\footnotesize
\setlength{\parskip}{0pt}
\tableofcontents
\endgroup
\newpage

\section{Introduction}\label{sec:introduction}
The quantum Stein lemma identifies the optimal asymptotic type-II error exponent for two i.i.d. quantum states $\rho,\sigma $ with their relative entropy \cite{HP,ON}. Namely, for every $\varepsilon\in (0,1)$
\[ \lim_{n\to \infty }\frac{1}{n}D_{H}^\varepsilon(\rho^{\otimes n}||\sigma^{\otimes n})=D(\rho||\sigma )\]
where $D_{H}^\varepsilon$ is the hypothesis testing relative entropy.
Brand\~ao and Plenio in \cite{BP} proposed its generalization to composite alternative hypotheses of convex families, called Generalized Quantum Stein Lemma (GQSL). It asserts that under certain assumptions on the convex composite hypotheses, the asymptotic worst case type-II error exponent is given by the regularized relative entropy with respect to alternative hypotheses. A gap in the original achievability proof was identified in \cite{Gap}. The theorem was subsequently established by Hayashi and Yamasaki \cite{HY} for finite-dimensional quantum systems and independently by Lami \cite{Lami}. In particular, Hayashi and Yamasaki removed the original assumptions of closure under partial trace and permutation invariance.

Here, we extend the Generalized Quantum Stein Lemma to arbitrary von Neumann algebras, under the assumption of the existence of a one-copy alternative state $\sigma$ with finite relative entropy $D(\rho||\sigma)<\infty$ from the null state $\rho$. The convex family $\mathcal{F}=(\mathcal{F}_n)_n$ are assumed to be convex and stable under tensor products, but not necessarily be norm-closed, permutation invariant or closed under
restrictions. The conclusion includes the error exponent and a strong converse.

\subsection{Setting and main result}\label{sec:theorem}
Let $\M$ be a von Neumann algebra, let $\Sn(\M)$ denote its normal state space, and fix a state $\rho\in\Sn(\M)$. Write $\M_n=\M^{\otimes n}$, with the tensor factors in their given order. Let $\mathcal{F}=(\F_n)_{n\ge 1}$ be a family of state sets $\F_n\subseteq\Sn(\M_n)$. A test  $T\in\M_n$ is an operator $0\le T\le\one$ that accepts the null state $\rhon$ with probability $\rhon(T)$. The type-I and worst-case type-II errors are
\begin{align}
 \alpha_n(T):=1-\rhon(T),\qquad
 B_n(T):=\sup_{\sigma\in\F_n}\sigma(T).\label{eq:worst-beta}
\end{align}
For $0<\varepsilon<1$, we define the hypothesis testing relative entropy
\begin{align}
 \beta_{n,\varepsilon}
 &:=\inf_{\substack{0\le T\le\one\\\rhon(T)\ge1-\varepsilon}}B_n(T),\qquad
 \Dht^\varepsilon(\rhon\|\F_n):=-\log\beta_{n,\varepsilon}.
 \label{eq:beta}
\end{align}
All logarithms are natural with convention $-\log0=+\infty$. We will make the following assumption on the family $\mathcal{F}=(\F_n)_{n\ge 1}$.
\begin{assumption}\label{ass:main}
The family $(\F_n)_{n\ge1}$ has the following properties.
\begin{enumerate}[label=\textup{(A\arabic*)},leftmargin=3.5em]
 \item\label{ass:convex} Each $\F_n$ is nonempty and convex.
 \item\label{ass:tensor} It is stable under ordered tensor products:
 \begin{equation}
 \sigma\in\F_n,\quad\omega\in\F_m
 \quad\Longrightarrow\quad\sigma\otimes\omega\in\F_{n+m}.
 \label{eq:tensor-stability}
 \end{equation}
 \item\label{ass:reference} There is a state $\tau\in\F_1$ with
 \begin{equation}
 C:=D(\rho\|\tau)<\infty.
 \label{eq:finite-reference}
 \end{equation}
\end{enumerate}
\end{assumption}
Note that the finite relative entropy condition (A3) is relative to the null state $\rho$. It is not a uniform condition over all normal states. We do not assume the faithfulness of $\rho$ or $\tau$, or $\sigma$-finiteness of $\M$. Also, we do \emph{not} assume norm-closedness, closure under restrictions to tensor factors, permutation invariance of the family $\mathcal{F}=(\F_n)_{n\ge 1}$.

In the following, $D$ denotes Araki relative entropy (see Subsection~\ref{subsec:entropy-definition}). 
For normal states $\omega,\sigma$ and $t>0$, recall the hockey-stick divergences
\begin{equation}
 \hock_t(\omega\|\sigma):=(\omega-t\sigma)_+(\one),
 \label{eq:hock-intro}
\end{equation}
where the positive part is from the Jordan decomposition in the predual $\M_*$. For $r\in\mathbb R$ and $n \ge 1$, define the composite hockey-stick quantity
\begin{equation}
 \Q_n(r):=\inf_{\sigma\in\F_n}
           \hock_{e^{nr}}(\rhon\|\sigma).
 \label{eq:Qn}
\end{equation}

\begin{theorem}[Generalized quantum Stein lemma]\label{thm:main}
Under Assumption~\ref{ass:main}, the limit
\begin{equation}
 R:=\lim_{n\to\infty}\frac{1}{n}\inf_{\sigma\in\F_n}D(\rhon\|\sigma)
   =\inf_{n\ge1}\frac{1}{n}\inf_{\sigma\in\F_n}D(\rhon\|\sigma)
 \quad\text{exists and $R\in [0,C]$} 
 \label{eq:rate}
\end{equation}
The composite hockey-stick quantities have the sharp threshold
\begin{equation}
 \lim_{n\to\infty}\Q_n(r)=
 \begin{cases}
  1,&r<R,\\
  0,&r>R.
 \end{cases}
 \label{eq:threshold}
\end{equation}
 Consequently, for every $0<\varepsilon<1$,
\begin{equation}
 \boxed{\displaystyle
 \lim_{n\to\infty}\frac1n\Dht^\varepsilon(\rhon\|\F_n)=\lim_{n\to\infty}\frac{1}{n}\inf_{\sigma\in\F_n}D(\rhon\|\sigma)=R,}
 \label{eq:main-limit}
\end{equation}
where each $r<R$ is achievable by tests with vanishing type-I error $\alpha_n(T)\to 0$. Moreover, we have the strong converse that if a sequence of tests $(T_n)$ satisfies
\begin{equation}
 \liminf_{n\to\infty}-\frac1n\log B_n(T_n)>R,
 \label{eq:SC-premise}
\end{equation}
then the type-I error $\alpha_n(T)\to 1$.
\end{theorem}

\subsection{Relation to previous work and organization}
For finite dimensional system $\M=B(\mathbb C^d)$, a full-rank reference state guarantees the finite-entropy assumption, so the above conclusion recovers the theorem of Hayashi--Yamasaki \cite{HY}. The extension here allows infinite-dimensional $B(H)$ as well as general von Neumann algebras of type II and~III. It extends the recent subalgebra result of Bergh, Gao and Rahaman \cite{BGR} and the von Neumann algebraic binary-testing and complements equipartition results in \cite{JOPS,PW,FGR,HM}. 

The proof uses Frenkel's integral representation \cite{Frenkel} of relative entropy, recently established in von Neumann algebras \cite{CFLS,KSLC}. It is combined with modular testing estimates, $L_1$-convergence of fixed-block testing, convex mixing and minimax argument. Section~\ref{sec:analytic} collects the definitions and preliminary results. Section~\ref{sec:proof} proves Theorem~\ref{thm:main}. 

{\bf \noindent Acknowledgements.}
The author thank Kun Fang, Bjarne Bergh and Mizanur Rahaman for helpful discussions on this problem. 
The author is supported in part by the National Natural Science Foundation of China (grant no.~12401163) and the Department of Science and Technology of Hubei Province (project nos.~2025EHA041 and~2025AFA044).

\section{Preliminaries}\label{sec:analytic}
We collect the definitions, analytic estimates, and block constructions used in the proof.

\subsection{Normal states and relative entropy}\label{subsec:entropy-definition} Let $\M$ be a von Neumann algebra with predual $\M_*$.
The normal state space of $\M$ is
\[
 \Sn(\M)=\{\omega\in\M_*:\omega(\one)=1,\ \omega\ge0\}.
\]
All tensor products of von Neumann algebras are spatial, and the identification $\M_{n+m}=\M_n\bar\otimes\M_m$ preserves the order of the factors. 
Recall that in the Stein exponent \eref{eq:beta},
\begin{align}
 \beta_{n,\varepsilon}
 &:=\inf_{\substack{0\le T\le\one\\\rhon(T)\ge1-\varepsilon}}B_n(T),\qquad
 \Dht^\varepsilon(\rhon\|\F_n):=-\log\beta_{n,\varepsilon}.
\end{align}
a test $T$ must work for every $\sigma\in\F_n$; the test cannot depend on the alternative $\sigma \in \F_n$.

For two normal states $\omega$ and $\sigma$, Araki relative entropy \cite{araki,araki2} is defined by
\begin{align}
  D(\omega\|\sigma)
  :=
    -\langle\xi_\omega,
       (\log\Delta_{\sigma,\omega})\xi_\omega\rangle,   \label{eq:araki}
\end{align}
provided that the supports $s(\omega)\leq s(\sigma)$, and $+\infty$ otherwise. Here $\xi_\omega$ is the unique vector in the positive cone $\mathcal P$ of the standard form $(\mathcal M,\mathcal H,J,\mathcal P)$ state $\omega$, $\Delta_{\sigma,\omega}$ is the relative modular operator. More specifically, let $\supp(\omega)$ be the support projection of $\omega$ in $\M$ and $r(\omega)=J\supp(\omega) J$ the support projection in the commutant $\M'=J\M J$. Using the relative Tomita convention of \cite[Definition~3.1 and Lemma~3.2, pp.~21--22]{CFLS} and \cite[Section~5.2, Proposition~5.2, p.~30]{JOPS}, the anti-linear operator $S_{\sigma,\omega}$ is the closure of
\[
 a\xi_\omega+z\longmapsto \supp(\omega) a^*\xi_\sigma,
 \qquad a\in\M,\quad z\in(\one-r(\omega))H,
\]
and $$\Delta_{\sigma,\omega}=S_{\sigma,\omega}^*S_{\sigma,\omega} \qquad \text{ and }\qquad S_{\sigma,\omega}=J\Delta_{\sigma,\omega}^{1/2}.$$ The support of $\Delta_{\sigma,\omega}$ is $\supp(\omega)r(\omega)$.
These definitions apply to arbitrary positive normal functionals without full support. 

 The logarithmic expectation in \eqref{eq:araki} is understood as the 
spectral integral
\[
  -\langle\xi_\omega,
    (\log\Delta_{\sigma,\omega})\xi_\omega\rangle
  :=
  \int_{(0,\infty)}(-\log t)\,
  d\langle\xi_\omega,E_{\sigma,\omega}(t)\xi_\omega\rangle,
\]
where $E_{\sigma,\omega}=1_{[0,t]}(\Delta_{\sigma,\omega})$ is the spectral resolution of
$\Delta_{\sigma,\omega}$. This integral may equal $+\infty$
even when $s(\omega)\leq s(\sigma)$. Its value is independent of the chosen vector representation and corresponding spatial derivative. The relative modular operator and its support conventions are recalled in Subsection~\ref{subsec:modular}. In $B(H)$ it coincides
with the usual Umegaki relative entropy of density operators,
\[ D(\omega\|\sigma)=\Tr\bigl(\omega(\log\omega-\log\sigma)\bigr),\]
understood by its standard extended-value definition.

\subsection{Hockey-stick divergences}
For a hermitian normal functional $\eta\in \M_*$, the Jordan decomposition $\eta=\eta_+-\eta_-$ satisfies
\begin{equation}
 \eta_+(\one)=\max_{0\le T\le\one}\eta(T)
             =\frac{\norm{\eta}+\eta(\one)}2.
 \label{eq:jordan-variational}
\end{equation}
Indeed, $\eta(T)\le\eta_+(T)\le\eta_+(\one)$, and equality holds at the support projection of $\eta_+$. Thus
\begin{equation}
 \hock_t(\omega\|\sigma)
 =\max_{0\le T\le\one}\{\omega(T)-t\sigma(T)\}
 =\frac{\norm{\omega-t\sigma}+1-t}{2}.
 \label{eq:hock-variational}
\end{equation}

\begin{lemma}[Basic bounds]\label{lem:hock}
For two normal states $\omega$ and $\sigma$, $0\le\hock_t(\omega\|\sigma)\le1$. It is nonincreasing in $t$, and
\begin{align}
 \sigma\ge c\theta\quad&\Longrightarrow\quad
 \hock_t(\omega\|\sigma)\le\hock_{ct}(\omega\|\theta),
 \qquad c>0,\label{eq:hock-order}\\
 \big|\hock_t(\omega\|\sigma)-\hock_t(\omega'\|\sigma')\big|
 &\le \frac{1}{2}\norm{\omega-\omega'}+\frac{t}{2}\norm{\sigma-\sigma'}.
 \label{eq:hock-continuity}
\end{align}
In particular, $t\mapsto\hock_t(\omega\|\sigma)$ is continuous.
\end{lemma}
\begin{proof} The proof is the same as the finite dimensional case \cite[Lemma~2.1 and Corollary~5.7]{HT24}. We include an argument for completeness in the von Neumann algebra setting.
The lower bound follows by taking $T=0$ in \eref{eq:hock-variational}, and the upper bound follows from $\omega(T)\le1$. Increasing the second functional lowers the expression for every $T$. The continuity bound \eref{eq:hock-continuity} follows from the norm expression \eqref{eq:hock-variational}.
\end{proof}

We will use the following Frenkel's integral representation of relative entropy \cite[Theorem~6]{Frenkel}. This formula is extended to infinite-dimensional $B(H)$ by Jen\v{c}ov\'a \cite[Theorem~1]{Jencova}, and recently established on von Neumann algebras by Correa da Silva, Fr\"ob, Lechner and Sangaletti \cite[Theorem~3.19]{CFLS} and independently by Koßmann,  Schwonnek, Liu, and Cheng in \cite[Theorem~6.5]{KSLC}.

\begin{theorem}[Integral representation of relative entropy]\label{thm:integral}
For two normal states $\omega,\sigma$ on an arbitrary von Neumann algebra,
\begin{equation}
 D(\omega\|\sigma)
 =\int_1^\infty\hock_t(\omega\|\sigma)\,\frac{\dd t}{t}
  +\int_1^\infty\hock_t(\sigma\|\omega)\,\frac{\dd t}{t^2}.
 \label{eq:entropy-integral}
\end{equation}
\end{theorem}

Since $\hock_t(\sigma\|\omega)\le1$, the second integral is between zero and one. Substituting $t=e^{nx}$ into the first integral gives the following useful form.

\begin{corollary}\label{cor:area}
For normal states $\omega_n,\sigma_n$ on any von Neumann algebra,
\begin{equation}
 \frac1nD(\omega_n\|\sigma_n)
 =\int_0^\infty\hock_{e^{nx}}(\omega_n\|\sigma_n)\dd x+\eta_n, 
 \qquad \text{ with }  \qquad 0\le\eta_n\le\frac1n.
 \label{eq:area}
\end{equation}
\end{corollary}


\subsection{A modular testing bound}\label{subsec:modular}
Denote the spectral measure of $\Delta_{\sigma,\omega}$ in $\xi_\omega$ as
\[ \mu_{\omega,\sigma}([0,t])=d\langle\xi_\omega,E_{\sigma,\omega}(t)\xi_\omega\rangle\]
The derivative needed below is proved directly, which does not require the additional support assumptions in \cite[Proposition~5.4]{JOPS}.

\begin{lemma}[Finite entropy gives the required modular derivative]\label{lem:moment}
Let $\omega,\sigma$ be normal states with $D(\omega\|\sigma)<\infty$. For $0<s<1$, define
\begin{equation}
 q_{\omega,\sigma}(s)
 :=\norm{\Delta_{\sigma,\omega}^{s/2}\xi_\omega}^{2}
 =\int_{(0,\infty)}\lambda^s\dd\mu_{\omega,\sigma}(\lambda).
 \label{eq:moment}
\end{equation}
Then $0<q_{\omega,\sigma}(s)\le1$ and
\begin{equation}
 \lim_{s\to 0^+}-\frac1s\log q_{\omega,\sigma}(s)
 =D(\omega\|\sigma).
 \label{eq:moment-limit}
\end{equation}
The expression in \eref{eq:moment} is interpreted as a quadratic-form value which does not require $\xi_\omega$ to lie in the domain of $\Delta_{\sigma,\omega}^{s}$. 
\end{lemma}
\begin{proof}
Write $\mu=\mu_{\omega,\sigma}$. Finite relative entropy implies
$\supp(\omega)\le\supp(\sigma)$, so $\mu(\{0\})=0$ by the kernel formula in \cite[Proposition~5.2(2)]{JOPS}. Since $\omega$ is a state, $\mu$ has total mass one. The identity $\Delta_{\sigma,\omega}^{1/2}\xi_\omega=r(\omega)\xi_\sigma$ from \cite[Proposition~5.2(4)]{JOPS}, together with $J\xi_\sigma=\xi_\sigma$, gives
\begin{align}
 \int\lambda\dd\mu(\lambda)
 =\norm{\Delta_{\sigma,\omega}^{1/2}\xi_\omega}^{2}
 =\norm{r(\omega)\xi_\sigma}^{2}
 =\norm{\supp(\omega)\xi_\sigma}^{2}
 =\sigma(\supp(\omega))\le1.
 \label{eq:first-moment}
\end{align} Jensen's inequality for the concave function $\lambda\mapsto \lambda^s$ now gives $q(s)\le1$, and positivity follows from $\mu((0,\infty))=1$. Araki's spectral formula is
\begin{equation}
 D(\omega\|\sigma)=-\int\log\lambda\dd\mu(\lambda).
 \label{eq:araki-spectral}
\end{equation}
The positive part of $\log\lambda$ is integrable by \eref{eq:first-moment}. Finiteness in \eref{eq:araki-spectral} therefore also makes its negative part integrable. Consequently,
$\int|\log\lambda|\dd\mu<\infty$.

For $0<s\le1/2$ and $0<\lambda\le1$,
\[
 \left|\frac{\lambda^s-1}{s}\right|\le-\log\lambda.
\]
For $\lambda\ge1$, the mean value theorem for $s\in [0,1/2]$ yields
\[
 0\le\frac{\lambda^s-1}{s}
 \le\lambda^{1/2}\log\lambda
 \le\frac2e\lambda.
\]
These bounds are integrable and independent of $s$. Dominated convergence gives
\[ \qquad q(s)\longrightarrow1, \qquad 
 \lim_{s \to 0^+}\frac{q(s)-1}{s}
 =\int\log\lambda\dd\mu=-D(\omega\|\sigma),
\]
The logarithm is differentiable at one, so \eref{eq:moment-limit} follows from the chain rule.
\end{proof}
The following is from \cite[Proposition~3.8, p.~25]{CFLS}. We include the argument for completeness.
\begin{lemma}[Modular testing bound]\label{lem:modular-bound}
Let $\omega,\sigma$ be two normal states and $q_{\omega,\sigma}(s)$ be the quadratic-form as above. For $t>0$ and $0<s<1$,
\begin{equation}
 \hock_t(\omega\|\sigma)
 \ge1-t^s q_{\omega,\sigma}(s).
 \label{eq:modular-bound}
\end{equation}
\end{lemma}
\begin{proof}
Ogata's generalized Powers--St\o rmer inequality \cite{Ogata} states
\[
 \frac{\varphi(\one)+\psi(\one)-\norm{\varphi-\psi}}2
 \le\norm{\Delta_{\psi,\varphi}^{s/2}\xi_\varphi}^{2}.
\]
Taking $\varphi=\omega$ and $\psi=t\sigma$,
the assertions follows from the scaling rule
$\Delta_{t\sigma,\omega}=t\Delta_{\sigma,\omega}$, and \eref{eq:hock-variational}. 
\end{proof}

\subsection{Tensorization and Regularization}\label{subsec:tensor-modular}
We will also use the following tensorization property. 
For two pairs of normal states $(\sigma_1,\omega_1), (\sigma_2,\omega_2)$  on two von Neumann algebras,
one has in the product standard form
\begin{equation}
 \Delta_{\sigma_1\otimes\sigma_2,\,\omega_1\otimes\omega_2}
 =\Delta_{\sigma_1,\omega_1}\otimes\Delta_{\sigma_2,\omega_2},
 \label{eq:modular-tensor}
\end{equation}
with corresponding supports. Consequently, for $0<s<1$, we have
\begin{equation}
 q_{\omega_1\otimes\omega_2,\,\sigma_1\otimes\sigma_2}(s)
 =q_{\omega_1,\sigma_1}(s)q_{\omega_2,\sigma_2}(s).
 \label{eq:q-tensor}
\end{equation}
where the left-hand side is defined by joint spectral calculus. Similarly, for relative entropy
\begin{equation}
 D(\omega_1\otimes\omega_2\|\sigma_1\otimes\sigma_2)
 =D(\omega_1\|\sigma_1)+D(\omega_2\|\sigma_2).
 \label{eq:D-tensor}
\end{equation}
given both terms on the right-hand side are finite. We will only consider finite-entropy product pairs in what follows.

For a normal state $\omega$ and a subset $\mathcal{F}$, denote 
\[ D(\omega|| \F)=\inf_{\sigma\in\F}D(\omega\|\sigma)\]
The existence of the relative entropy rate follows from Fekete's subadditive lemma
\cite{Fekete} as in finite dimensions (see
\cite[Sec.~II, following Eq.~(10)]{BP} and
\cite[Supplementary Information, Lemma~S4]{HY}).
\begin{lemma}[Existence of the entropy rate]\label{lem:rate}
Let $\rho\in S(\M)$ be a normal state and $\mathcal{F}_n\subset \mathcal{S}(\M^{\otimes n})$ be a family satisfying $\F_n\otimes \F_m\subset \F_{n+m}$. Then
\begin{align*}
D(\rho^{\otimes (n+m)}|| \F_{n+m})\le D(\rhon|| \F_n)+D(\rho^{\otimes m}|| \F_m). \end{align*} Consequently, if $D(\rho|| \F_1)\le C$, then 
\begin{equation}
 R:=\lim_{n\to\infty}\frac{1}{n}D(\rhon|| \F_n)
   =\inf_{n\ge1}\frac{1}{n}D(\rhon|| \F_n)
 \quad\text{exists and $R\le C$} 
 \label{eq:rate2}
\end{equation}
Moreover, for every $\delta>0$ there are $k\ge1$ and $\sigma_k\in\F_k$ such that
\begin{equation}
 R\le\frac1kD(\rho^{\otimes k}\|\sigma_k)<R+\delta
 \label{eq:near-block}
\end{equation}
\end{lemma}

\section{Proof of the main theorem}\label{sec:proof}
We first prove convergence of the fixed-block testing, then establish the convex mixing estimate. Minimax argument yields uniform composite tests, from which the fixed-error limit and strong converse follow.

\subsection{The fixed-block testing profile converges in \texorpdfstring{$L^1$}{L1}}\label{sec:profile}

Fix $k\ge1$ and a state $\sigma_k\in\F_k$ with finite entropy rate $b=\frac{1}{k}D(\rho^{\otimes k}|| \sigma_k)<\infty$ for the moment; it need not yet be nearly optimal. For every $n$, write
\begin{equation}
 n=mk+j,\qquad 0\le j<k,
 \qquad
 \theta_n:=\sigma_k^{\otimes m}\otimes\tau^{\otimes j}.
 \label{eq:theta}
\end{equation}
Tensor stability ensures $\theta_n\in\F_n$. By tensor additivity, we have
\begin{equation}
 d_n:=\frac1nD(\rhon\|\theta_n)
 =\frac{mkb+jC}{n}\longrightarrow b.
 \label{eq:dn}
\end{equation}

Define
\begin{equation}
 g_n(x):=\hock_{e^{nx}}(\rhon\|\theta_n),\qquad x\ge0.
 \label{eq:gn}
\end{equation}
Each $g_n$ is nonincreasing, measurable, and takes values in $[0,1]$.
The key ingredient in our analysis is the following proposition.
\begin{proposition}\label{prop:profile}
For the sequence \eref{eq:theta}, with finite block rate $b=\frac{1}{k}D(\rho^{\otimes k}|| \sigma_k)<\infty$,
\begin{equation}
 \int_0^\infty
 \left|g_n(x)-\one_{[0,b)}(x)\right|\dd x\longrightarrow0.
 \label{eq:L1}
\end{equation}
In particular,
\begin{equation}
 e_n:=\int_b^\infty g_n(x)\dd x\longrightarrow0,
 \label{eq:tail}
\end{equation}
and for every $s>b$,
\begin{equation}
 0\le g_n(s)\le\frac{e_n}{s-b}\longrightarrow0.
 \label{eq:profile-upper}
\end{equation}
\end{proposition}
\begin{proof}
We separate the lower part of the integral from the tail.
Assume $b>0$ and fix $0\le x<b$. Abbreviate
\[
 q_k(u):=q_{\rho^{\otimes k},\sigma_k}(u),\qquad
 b_u:=-\frac{1}{ku}\log q_k(u),\qquad0<u<1.
\]
Lemma~\ref{lem:moment} gives $b_u\to b$ as $u\to 0^+$. Choose and then fix $u$ such that $b_u>x$.

For $n\ge k$, restricting the maximization in \eref{eq:hock-variational} to those of the form $T\otimes\one^{\otimes j}$ gives
\begin{align}
 g_n(x)
 \ge\hock_{e^{nx}}\bigl((\rho^{\otimes k})^{\otimes m}
                     \big\|\sigma_k^{\otimes m}\bigr)\ge1-e^{nux}q_k(u)^m
 =1-\exp\!\left(nux-mku b_u\right).
 \label{eq:profile-lower}
\end{align}
The second inequality uses Lemma~\ref{lem:modular-bound} and \eref{eq:q-tensor}. Since $mk/n\to1$, the exponent divided by $n$  
$$(nux-mku b_u)/n \to u(x-b_u)<0 .$$
Then the exponential term therefore tends to zero, proving $g_n(x)\to1$ for $0\le x<b$.


By Corollary~\ref{cor:area} and \eref{eq:dn},
\begin{equation}
 \int_0^\infty g_n(x)\dd x=d_n-\eta_n\longrightarrow b,
 \qquad0\le\eta_n\le\frac1n.
 \label{eq:profile-total}
\end{equation}

For $b>0$, bounded convergence on the \emph{finite} interval $[0,b]$ gives
\[
 \int_0^b g_n(x)\dd x\longrightarrow b.
\]
Subtraction from \eref{eq:profile-total} gives \eref{eq:tail}, and
\[
 \int_0^\infty|g_n-\one_{[0,b)}|\dd x
 =b-\int_0^b g_n(x)\dd x+\int_b^\infty g_n(x)\dd x
 \longrightarrow0.
\]
If $b=0$, \eref{eq:profile-total} directly says $\int_0^\infty g_n\to0$, which proves both \eref{eq:L1} and \eref{eq:tail} in that case. For $s>b$, monotonicity gives
\[
 (s-b)g_n(s)\le\int_b^s g_n(x)\dd x\le e_n.
\]
This proves \eref{eq:profile-upper} and finishes the proof.
\end{proof}

\begin{remark}[Upper-tail convergence]
The lower testing bound \eref{eq:modular-bound} establishes convergence below $b$. The entropy integral supplies convergence of the total area, forcing the upper tail to vanish. These two ingredients yield the converse estimates used below.
\end{remark}

\subsection{The convex mixing estimate}\label{sec:mixing}

\begin{lemma}[Two-comparator mixing inequality]\label{lem:mixing}
Let $\zeta,\omega,\theta$ be normal states with $D(\zeta\|\theta)<\infty$. Put $\nu=(\omega+\theta)/2$. Then, for $1\le u\le v<\infty$,
\begin{align}
 D(\zeta\|\nu)
 \le{}1+\log(2u)
 +\hock_u(\zeta\|\omega)\log\frac vu+\int_v^\infty\hock_t(\zeta\|\theta)\,\frac{\dd t}{t}.
 \label{eq:mixing}
\end{align}
In particular, $D(\zeta\|\nu)$ is finite.
\end{lemma}
\begin{proof}
The order relation $\nu\ge\omega/2$ and $\nu\ge\theta/2$ imply
\begin{equation}
 \hock_t(\zeta\|\nu)
 \le\hock_{t/2}(\zeta\|\omega),
 \qquad
 \hock_t(\zeta\|\nu)
 \le\hock_{t/2}(\zeta\|\theta).
 \label{eq:mixing-order}
\end{equation}
Apply the integral identity \eref{eq:entropy-integral}, the second integral $\displaystyle\int_1^\infty\hock_t(\nu\|\zeta)\,\frac{\dd t}{t^2}$ is at most $1$. Split the first integral of $\hock_t(\zeta\|\nu)$ as
\[
 \int_1^\infty=\int_1^{2u}+\int_{2u}^{2v}+\int_{2v}^\infty.
\]
On the first interval, use $\hock_t(\zeta\|\nu)\le1$, obtaining $\log(2u)$. On the second, \eref{eq:mixing-order} and $t\ge 2u$ give the bound $\hock_u(\zeta\|\omega)$, so its contribution is at most
$\hock_u(\zeta\|\omega)\log(v/u)$. On the third, use the second bound in \eref{eq:mixing-order} and change variables $t/2\mapsto t$. This gives the last integral in \eqref{eq:mixing}. That integral is finite by $v\ge1$ and $D(\zeta\|\theta)<\infty$. All bounds therefore yield a finite right-hand side, which also proves finiteness of the left-hand side.
\end{proof}

The state $\omega$ controls the middle interval of the testing threshold, while $\theta$ controls the tail. Both bounds apply to the same mixture $\nu$. If both comparison states belong to a convex alternative family, then so does $\nu$, and the estimate bounds the relative entropy at an admissible alternative.

\subsection{The composite threshold}\label{sec:composite}
In this part, we prove the the sharp threshold
\begin{equation}
 \lim_{n\to\infty}\Q_n(r)=
 \begin{cases}
  1,&r<R,\\
  0,&r>R.
 \end{cases}
 \label{eq:threshold2}
\end{equation}
for the composite hockey-stick quantities 
\begin{equation}
 \Q_n(r):=\inf_{\sigma\in\F_n}
           \hock_{e^{nr}}(\rhon\|\sigma).
 \label{eq:Qn2}
\end{equation}
We start with a minimax lemma.
\begin{lemma}[Uniform testing by minimax]\label{lem:minimax}
Let $\mathcal A$ be a von Neumann algebra, $\omega\in\Sn(\mathcal A)$, $\mathcal C\subseteq\Sn(\mathcal A)$ be a nonempty convex set, and $t>0$. Then
\begin{align}
 \inf_{\sigma\in\mathcal C}\hock_t(\omega\|\sigma)
 &=\max_{0\le T\le\one}
 \left\{\omega(T)-t\sup_{\sigma\in\mathcal C}\sigma(T)\right\}.
 \label{eq:minimax}
\end{align}
Neither compactness nor closedness of $\mathcal C$ is required.
\end{lemma}
\begin{proof}
View $\mathcal A_{\mathrm{sa}}$ as a real locally convex space with its ultraweak topology. The test set
\[
 K=\{T\in\mathcal A_{\mathrm{sa}}:0\le T\le\one\}
\]
is convex and weak$^*$-compact. Give $\mathcal C$ the subspace topology from the real Banach space $\mathcal A_{*,\mathrm{sa}}$ with its norm.

The real-valued function
\[
 f(T,\sigma)=\omega(T)-t\sigma(T)
\]
is affine in each variable. For fixed $\sigma$, it is weak$^*$-continuous in $T$ because $\omega,\sigma$ are normal. For fixed $T$, it is norm-continuous in $\sigma$. Sion's minimax theorem, in the form with only one compact side \cite[p.~174]{Sion}, therefore gives
\[
 \inf_{\sigma\in\mathcal C}\sup_{T\in K}f(T,\sigma)
 =\sup_{T\in K}\inf_{\sigma\in\mathcal C}f(T,\sigma).
\]
The left-hand side is the left-hand side of \eref{eq:minimax}. For the right-hand side, $T\mapsto\inf_{\sigma\in\mathcal C}f(T,\sigma)$ is upper semicontinuous, being an infimum of continuous functions, and it is bounded. Hence its supremum on $K$ is attained. Finally,
\[
 \inf_{\sigma\in\mathcal C}f(T,\sigma)
 =\omega(T)-t\sup_{\sigma\in\mathcal C}\sigma(T).
\]
This proves the claim.
\end{proof}

Applying the lemma to $\mathcal A=\M_n$, $\mathcal C=\F_n$, $t=e^{nr}$ gives
\begin{equation}
 1-\Q_n(r)
 =\min_{0\le T\le\one}
   \{\alpha_n(T)+e^{nr}B_n(T)\}.
 \label{eq:bayesian-minimax}
\end{equation}
Thus a hockey-stick value close to one yields a \emph{single} good test for the entire family.

We now prove the lower threshold for $\Q_n(r)$ using the convex-set argument. Recall that the threshold value is
\begin{equation}
 R:=\lim_{n\to\infty}\frac{1}{n}D(\rhon|| \F_n)
   =\inf_{n\ge1}\frac{1}{n}D(\rhon|| \F_n)
 \quad\text{exists and $R\le C$} 
 \label{eq:rate3}
\end{equation}
\begin{proposition}\label{prop:direct}
For every $r<R$, $\Q_n(r)\to1$.
\end{proposition}
\begin{proof}
For $r<0$, the test $T=\one$ in \eref{eq:hock-variational} gives
\[
 \Q_n(r)\ge1-e^{nr}\longrightarrow1.
\]
It remains to consider $0\le r<R$, which in particular implies $R>0$.

Fix a finite-entropy block comparator $\sigma_k\in\F_k$ and its rate $b=\frac{1}{k}D(\rho^{\otimes k}|| \sigma_k)<\infty$. Since $b\ge R>r$, we can apply Lemma~\ref{lem:mixing} with
\[
 \zeta=\rhon,\quad \theta=\theta_n,\quad
 u=e^{nr},\quad v=e^{nb}.
\]
For an \emph{arbitrary} $\omega_n\in\F_n$, convexity gives
\[
 \nu_n:=\tfrac12(\omega_n+\theta_n)\in\F_n.
\]
Consequently, using $D(\rhon|| \F_n)\le D(\rhon\|\nu_n)$, division of \eref{eq:mixing} by $n$ yields
\begin{align}
 \frac{1}{n}D(\rhon|| \F_n)
 &\le r+(b-r)\hock_{e^{nr}}(\rhon\|\omega_n)
       +\frac{1+\log2}{n}+e_n,
 \label{eq:central-estimate}
\end{align}
where the change of variables $t=e^{nx}$ gives exactly
\[
 \frac1n\int_{e^{nb}}^\infty
       \hock_t(\rhon\|\theta_n)\,\frac{\dd t}{t}
 =\int_b^\infty g_n(x)\dd x=e_n\longrightarrow0.
\]
The error terms in \eref{eq:central-estimate} depend on the chosen block comparator, but \emph{not} on $\omega_n$.

Rearrange \eref{eq:central-estimate} for each $\omega_n$ and then take its infimum:
\begin{equation}
 \Q_n(r)\ge
 \frac{\frac{1}{n}D(\rhon|| \F_n)-r-(1+\log2)/n-e_n}{b-r}.
 \label{eq:Q-lower-finite}
\end{equation}
For this fixed $k,\sigma_k$, let $n\to\infty$. Lemma~\ref{lem:rate} and Proposition~\ref{prop:profile} give
\begin{equation}
 \liminf_{n\to\infty}\Q_n(r)\ge\frac{R-r}{b-r}.
 \label{eq:Q-liminf}
\end{equation}
For every $\delta>0$, Lemma~\ref{lem:rate} provides a block rate $b<R+\delta$. Therefore
\[
 \liminf_n\Q_n(r)\ge\frac{R-r}{R+\delta-r}.
\]
Now let $\delta\to 0$. Since $\Q_n(r)\le1$, this proves the claim.
\end{proof}

\begin{corollary}[One test for all alternatives]\label{cor:uniform-tests}
For each $r<R$ there are tests $T_n$ with
\begin{equation}
 \alpha_n(T_n)+e^{nr}B_n(T_n)=1-\Q_n(r)\longrightarrow0.
 \label{eq:uniform-tests}
\end{equation}
In particular, $\alpha_n(T_n)\to 0$ and
$B_n(T_n)\le e^{-nr}(1-\Q_n(r))\le e^{-nr}$.
\end{corollary}
\begin{proof}
Choose an attaining test in \eref{eq:bayesian-minimax} for each $n$, and apply Proposition~\ref{prop:direct}. Both terms on the left are nonnegative.
\end{proof}

The upper threshold follows easily from Proposition \ref{prop:profile}
\begin{proposition}\label{prop:upper}
For every $r>R$, $\Q_n(r)\to0$.
\end{proposition}
\begin{proof}
Choose a block comparator $\sigma_k$ with $b=\frac{1}{k}D(\rho^{\otimes k}|| \sigma_k)<r$ by Lemma~\ref{lem:rate}. As $\theta_n\in\F_n$, by \eref{eq:profile-upper}
\[
 0\le\Q_n(r)\le\hock_{e^{nr}}(\rhon\|\theta_n)
 =g_n(r)\longrightarrow0
\]
\end{proof}

\subsection{Stein exponent and strong converse}\label{sec:operational}
The Stein exponent and its strong converse now follows easily from the sharp threshold.\\ 

{\bf \noindent Achievability.} Fix $0<\varepsilon<1$. If $R>0$, take $0\le r<R$ and the tests from Corollary~\ref{cor:uniform-tests}. For all sufficiently large $n$, they satisfy $\alpha_n(T_n)\le\varepsilon$, so
\[
 \beta_{n,\varepsilon}\le B_n(T_n)\le e^{-nr}.
\]
Thus
\begin{equation}
 \liminf_n\frac1n\Dht^\varepsilon(\rhon\|\F_n)\ge r.
 \label{eq:fixed-direct}
\end{equation}
Let $r\to R$. We have the lower bound  \begin{equation}
 \liminf_n\frac1n\Dht^\varepsilon(\rhon\|\F_n)\ge R.
 \label{eq:fixed-direct2}
\end{equation}

{\bf \noindent Strong converse.} 
Fix $c>R$. By Lemma~\ref{lem:rate}, choose a finite-entropy
block comparator with rate $b<c$, and choose $s$ with $b<s<c$.
Let $\theta_n$ and $g_n$ be the corresponding comparator states
and testing profiles. For every test  $T$ on $\M_n$, the
variational formula gives
\begin{equation}
 \rhon(T)-e^{ns}\theta_n(T)\le g_n(s).
 \label{eq:pair-converse}
\end{equation}
Since $\theta_n\in\F_n$, every test satisfying
$B_n(T)\le e^{-nc}$ therefore obeys
\begin{equation}
 \rhon(T)
 \le g_n(s)+e^{ns}B_n(T)
 \le g_n(s)+e^{-n(c-s)}
 \longrightarrow0,
 \label{eq:strong-converse-proof}
\end{equation}
where $g_n(s)\to0$ follows from \eref{eq:profile-upper}.
The bound is uniform over all such tests.

Suppose now that a sequence $(T_n)$ satisfies
\eref{eq:SC-premise}. Choose $c$ strictly between $\displaystyle R<c<\liminf_{n\to\infty}-\frac1n\log B_n(T_n)$ in that condition. Then $B_n(T_n)\le e^{-nc}$ for all
sufficiently large $n$, so \eref{eq:strong-converse-proof}
gives $\rhon(T_n)\to0$, proving the strong converse.

To obtain the fixed-error rate upper bound, fix $0<\varepsilon<1$.
For every $c>R$, the right-hand side of
\eref{eq:strong-converse-proof} is smaller than $1-\varepsilon$
for all sufficiently large $n$. Hence every admissible test
with $\rhon(T)\ge1-\varepsilon$ must satisfy $B_n(T)>e^{-nc}$.
Taking the infimum over these tests gives
\[
 \beta_{n,\varepsilon}\ge e^{-nc},
 \qquad
 \frac1n\Dht^\varepsilon(\rhon\|\F_n)\le c
\]
for all sufficiently large $n$. Taking the limsup and then
letting $c\downarrow R$, we obtain
\begin{equation}
 \limsup_{n\to\infty}
 \frac1n\Dht^\varepsilon(\rhon\|\F_n)\le R.
 \label{eq:fixed-upper}
\end{equation}
Together with \eref{eq:fixed-direct}, this proves \eref{eq:main-limit} and completes the proof
of Theorem~\ref{thm:main}.

\section{Discussion}\label{sec:scope}

The comparison with \cite{HY} on assumptions is especially direct. Its finite-dimensional state theorem already uses only closed convex alternatives, tensor-product stability, and a full-rank one-copy reference. The contribution here is the extension of this testing conclusion to arbitrary von Neumann algebras, together with a proof based on integral representation of relative entropy. This includes infinite-dimensional $B(H)$ and algebras of types~II and~III; neither $\sigma$-finiteness nor faithfulness of $\rho$ or the reference $\tau$ is assumed. 

The finite reference condition $D(\rho||\tau)$ is only for the target state $\rho$. In finite dimensions a full-rank reference guarantees this condition for every target state, whereas in infinite dimensions faithfulness alone does not guarantee finite relative entropy. For example, take classical distributions $p_k=c/k^2$ and $t_k=e^{-k^2}/Z$ on $\mathbb{N}$. Both are faithful, but
\[
 D(p\|t)=\sum_k p_k\bigl(\log p_k+k^2+\log Z\bigr)=\infty,
\]
because $\sum_k p_k|\log p_k|<\infty$ whereas $\sum_k p_k k^2=\infty$. The theorem here does not address families for which no finite-entropy reference is available.

\end{document}